\documentclass[12pt,reqno]{amsart}

\newcommand\version{November 19, 2013}

\usepackage{amsmath,amsfonts,amsthm,amssymb,amsxtra}

\newtheorem{theorem}{Theorem}
\newtheorem{proposition}[theorem]{Proposition}
\newtheorem{lemma}[theorem]{Lemma}
\newtheorem{corollary}[theorem]{Corollary}

\theoremstyle{definition}

\theoremstyle{remark}

\renewcommand{\epsilon}{\varepsilon}

\renewcommand{\phi}{\varphi}
\newcommand{\R}{\mathbb{R}}

\DeclareMathOperator{\im}{Im}

\DeclareMathOperator{\re}{Re}
\DeclareMathOperator{\spec}{spec}

\begin{document}

\title[Polaron dynamics --- \version]{Dynamics of a strongly coupled polaron}

\author{Rupert L. Frank}
\address{Rupert L. Frank, Mathematics 253-37, Caltech, Pasadena, CA 91125, USA}
\email{rlfrank@math.princeton.edu}

\author{Benjamin Schlein}
\address{Benjamin Schlein, University of Bonn, Endenicher Allee 60, 53115 Bonn, Germany}
\email{benjamin.schlein@hcm.uni-bonn.de}

\begin{abstract}
We study the dynamics of large polarons described by the Fr\"ohlich Hamiltonian in the limit of strong coupling. The initial conditions are (perturbations of) product states of an electron wave function and a phonon coherent state, as suggested by Pekar. We show that, to leading order on the natural time scale of the problem, the phonon field is stationary and the electron moves according to an effective linear Schr\"odinger equation.
\end{abstract}

\maketitle

\renewcommand{\thefootnote}{${}$} \footnotetext{\copyright\, 2013 by the authors. This paper may be reproduced, in its entirety, for non-commercial purposes.\\
U.S. National Science Foundation grant PHY-1347399 (R.F.) and ERC Starting Grant MAQD-240518 (B.S.) are acknowledged.}

\section{Introduction and main result}

The polaron is a model for an electron interacting with the quantized optical modes of a polar crystal. A `large' (or `continuous') polaron is characterized by the fact that the spatial extension of this polaron is large compared to the spacing of the underlying lattice. It can be described, as derived by Fr\"ohlich \cite{Fr} in 1937, by the Hamiltonian
$$
H^{\text{F}}_\alpha = p^2 + \int_{\R^3} \frac{dk}{|k|} \left( e^{-ik\cdot x} a_k + e^{ik\cdot x} a_k^* \right) + \int_{\R^3} dk\, a_k^* a_k \,,
$$
acting in $L^2(\R^3)\otimes\mathcal F$. Here, $x$ and $p=-i\nabla_x$ are position and momentum of the electron, respectively, and $a_k^*$ and $a_k$ are creation and annihilation operators in the symmetric Fock space $\mathcal F$ over $L^2(\R^3)$, satisfying
\begin{equation}\label{eq:ccr}
[a_k,a_{k'}^*]=\alpha^{-2}\delta(k-k') \,,
\qquad [a_k,a_{k'}]=[a_k^*,a_{k'}^*]=0
\qquad\text{for all}\ k,k'\in\R^3 \,.
\end{equation}
Note the $\alpha$ dependence in the commutation relations. We have written the Hamiltonian in strong coupling units, which will be convenient for us. In the appendix we explain the change of variables and relate it to the more standard form of this Hamiltonian. In Section \ref{sec:prelim} we also discuss the precise definition of this Hamiltonian and its lower boundedness.

Through the commutation relations, the Hamiltonian $H^{\text{F}}_\alpha$ depends on a single non-negative parameter $\alpha>0$, and we are interested in the so-called `strong coupling regime' $\alpha\to\infty$. The ground state energy
$$
E^{\text{F}}_\alpha = \inf\spec H^{\text{F}}_\alpha
$$
has been studied extensively. While its behavior for small $\alpha$ was understood completely by the middle of the 1950s \cite{LePi,LeLoPi,Gu,Fe,LiYa} the strong coupling regime remained open for quite some time. Pekar \cite{Pe0,Pe} had produced an upper bound on $E^{\text{F}}_\alpha$ by using a trial state of the product form
\begin{equation}
\label{eq:pekar}
\Psi = \psi \otimes W(\alpha^2 \phi)\Omega \,,
\end{equation}
where $\psi\in H^1(\R^3)$ is the wave function of an electron and $W(\alpha^2 \phi)\Omega$ is a coherent state corresponding to a phonon field $\phi\in L^2(\R^3)$. More formally, $\Omega$ is the vacuum in $\mathcal F$ and $W(f)$ is the Weyl operator,
$$
W(f) = \exp( a^*(f) - a(f) ) \,.
$$
For each $f\in L^2(\R^3)$, $W(f)$ is a unitary operator in $\mathcal F$. The property of these operators that will be important for us is that
\begin{equation}
\label{eq:weyl}
W^*(f)a_k W(f) = a_k + \alpha^{-2} f(k) 
\qquad\text{and}\qquad
W^*(f)a_k^* W(f) = a_k^* + \alpha^{-2} \overline{f(k)} \,. 
\end{equation}
In particular, coherent states are eigenstates of annihilation operators,
\begin{equation}
\label{eq:coherent}
a_k W(f)\Omega = \alpha^{-2} f(k) W(f)\Omega \,.
\end{equation}

The $\alpha$ enters in \eqref{eq:pekar} so that for fixed $\psi$ and $\phi$, the expected energy is bounded (indeed, constant) with respect to $\alpha$. To see this, we compute using \eqref{eq:coherent}
\begin{equation}
\label{eq:pekarham}
\left\langle \psi \otimes W(\alpha^2 \phi)\Omega, H_\alpha^{\text{F}} \left(\psi \otimes W(\alpha^2 \phi)\Omega \right)\right \rangle_{L^2(\R^3)\otimes\mathcal F}
= \left\langle \psi, H_\phi \psi \right\rangle_{L^2(\R^3)}
\end{equation}
with the effective Schr\"odinger operator
\begin{equation}
\label{eq:hameff}
H_\phi = p^2 + V_\phi(x) + \|\phi\|^2
\end{equation}
acting in $L^2(\R^3)$. Here,
$$
V_\phi(x) = \int_{\R^3} \frac{dk}{|k|} \left( e^{-ik\cdot x} \phi(k) + e^{ik\cdot x} \overline{\phi(k)} \right) = 2^{3/2} \pi^{-1/2} \re \int_{\R^3} \frac{dx}{|x-x'|^2} \check\phi(x')
$$
and
$$
\check \phi(x) = (2\pi)^{-3/2} \int_{\R^3} dk\, e^{-ik\cdot x} \phi(k) \,.
$$

By minimizing \eqref{eq:pekarham} over all $\psi$ and $\phi$, Pekar obtained an upper bound on $E^{\text{F}}_\alpha$ which he expected to be asymptotically correct as $\alpha\to\infty$. A mathematically rigorous proof of this fact was only achieved in 1983 by Donsker and Varadhan \cite{DoVa} using large deviation theory; for an alternative proof, using operator theory, see \cite{LiTh}.

While the ground state energy $E^{\text{F}}_\alpha$ has be studied extensively, we are not aware of any rigorous study of the dynamics 
$
e^{iH^{\text{F}}_\alpha}\Psi \,.
$
This is our concern here. More precisely, we are interested in the dynamics in the strong coupling limit $\alpha\to\infty$ for initial date $\Psi$ of the product form \eqref{eq:pekar} suggested by Pekar. Here is a special case of our main result.

\begin{theorem}\label{mainintro}
Let $\phi\in L^2(\R^3)$ and $\alpha_0>0$. Then for all $\psi\in H^1(\R^3)$, all $\alpha\geq \alpha_0$ and all $t\in\R$, 
$$
\left\| e^{-iH^{\text{F}}_\alpha t}\left(\psi\otimes W(\alpha^2\phi)\Omega\right) - \left(e^{-iH_\phi t}\psi\right) \otimes W(\alpha^2\phi)\Omega \right\|^2 \leq 2 \alpha^{-2} \|\psi\|_{H^1(\R^3)}^2 \left(e^{C|t|}-1\right) \,,
$$
where $C$ depends only on $\alpha_0$ and an upper bound on $\|\phi\|$.
\end{theorem}

In other words, the evolution of a Pekar product state \eqref{eq:pekar} can be approximated by dynamics of the electron wave function $\psi$ \emph{only}, and this evolution is described by the 
Schr\"odinger operator $H_\phi$ in \eqref{eq:hameff} with the effective potential $V_\phi$ determined by $\phi$. The coherent state describing the phonon field is stationary. This approximation is valid for times $|t|\leq o(\ln\alpha)$ and, in particular, for times of order one. Our main result, Theorem \ref{main2}, states that this approximation is also valid for certain initial states close to $\psi\otimes W(\alpha^2\phi)\Omega$ in an appropriate sense.

In our opinion this result is not unsurprising, since in the physics literature the motion of a strongly coupled polaron is typically described by the \emph{non-linear} system of equations
$$
i\partial_t \psi = \left(-\Delta +V\right) \psi \,,
\qquad
\left( c^{-2} \partial_t^2 + 1 \right)\Delta V = 4\pi |\psi|^2 \,;
$$
see, for instance, \cite{LaPe,BeNiRuSo,DeAl}. Our main result corresponds, in some sense, to the case $c=0$. We leave it as an \emph{open problem} to find a regime in which $c>0$.

Let us now state a more general version of Theorem \ref{mainintro} which also allows deviations from an exact product structure. To formulate our assumptions on the initial state we introduce the number of particles operator
\begin{equation}
\label{eq:nop}
\mathcal N = \int_{\R^3} dk\, a_k^* a_k
\end{equation}
acting in $\mathcal F$. Note that, if $\xi=(\xi^{(0)},\xi^{(1)},\ldots)\in\mathcal F$, then
$$
\langle\xi, \mathcal N \xi\rangle = \alpha^{-2} \sum_{n=1}^\infty n \|\xi^{(n)}\|^2 \,.
$$
The factor $\alpha^{-2}$ on the right side comes from the $\alpha$-dependence of the canonical commutation relations.

Our main result reads as follows.

\begin{theorem}\label{main2}
Let $\phi\in L^2(\R^3)$ and $\alpha_0>0$. Assume that $\Psi\in L^2(\R^3)\otimes\mathcal F$ satisfies
\begin{equation}
\label{eq:apriori}
\|(p^2+\mathcal N+1)^{1/2}\Psi\|\leq M\,,\quad
\|(p^2+1)^{1/2}\mathcal N\Psi\|\leq M\alpha^{-2} \,.
\end{equation}
Then for all $\alpha\geq\alpha_0$ and all $t\in\R$, 
$$
\left\| e^{-iH^{\text{F}}_\alpha t}W(\alpha^2\phi)\Psi - e^{-iH_\phi t}W(\alpha^2\phi) \Psi\right\|^2 \leq 2 M^2 \alpha^{-2} \left(e^{C|t|}-1\right) \,,
$$
where $C$ depends only on $\alpha_0$ and an upper bound on $\|\phi\|$.
\end{theorem}

This implies, of course, Theorem \ref{mainintro} by taking $\Psi =\psi\otimes\Omega$. Since $\mathcal N\Omega=0$, the two conditions in \eqref{eq:apriori} are satisfied with $M=\|\psi\|_{H^1}$, provided $\psi\in H^1(\R^3)$.

There is nothing special about the constant $2$ in this theorem (or in Theorem \ref{mainintro}). It can be replaced by any constant greater than one.

We now describe the strategy of our proof. We first observe that, since $W(\alpha^2\phi)$ is unitary and commutes with $H_\phi$, we have
\begin{align*}
\left\| e^{-iH^{\text{F}}_\alpha t}W(\alpha^2\phi)\Psi - e^{-iH_\phi t}W(\alpha^2\phi) \Psi\right\|^2
& = \left\| W^*(\alpha^2\phi) e^{-iH^{\text{F}}_\alpha t}W(\alpha^2\phi)\Psi - e^{-iH_\phi t} \Psi\right\|^2 \\
& = \left\| e^{-iW^*(\alpha^2\phi) H^{\text{F}}_\alpha W(\alpha^2\phi) t} \Psi - e^{-iH_\phi t} \Psi\right\|^2 \,.
\end{align*}
Moreover, a short computation based on \eqref{eq:weyl}, shows that
\begin{align*}
W^*(\alpha^2\phi) H^{\text{F}}_\alpha W(\alpha^2\phi) & = H_\phi + \int_{\R^3} dk\, a_k^* a_k + a(\phi)+ a^*(\phi) + \int_{\R^3} \frac{dk}{|k|}\left( e^{-ik\cdot x} a_k + e^{ik\cdot x} a_k^* \right) \\
& =: H \,.
\end{align*}
(Here, for the sake of simplicity, we do not indicate the dependence of $H$ on $\alpha$ and $\phi$.) In Section \ref{sec:prelim} we shall show that $H$, and therefore $H^{\text{F}}_\alpha$ as well, are lower semi-bounded operators in $L^2(\R^3)\otimes\mathcal F$. Since $|k|^{-1}\not\in L^2(\R^3)$, this is not completely obvious.

These manipulations have reduced the proof of Theorem \ref{main2} to the proof of the bound
\begin{equation}
\label{eq:mainequiv}
\left\| e^{-iH t}\Psi - e^{-iH_\phi t} \Psi\right\|^2 \leq 2 M^2 \alpha^{-2} \left(e^{C|t|}-1\right)
\end{equation}
with $C$ depending only on $\alpha_0$ and an upper bound on $\|\phi\|$. We shall prove (\ref{eq:mainequiv}) using a Gronwall-type argument, as explained in Proposition \ref{mainprop}.


\section{Form boundedness and energy conservation}\label{sec:prelim}

\subsection{The operator $H_\phi$}

Let $\phi\in L^2(\R^3)$. We want to argue that the potential $\check{\phi}*|x|^{-2}$ is infinitesimally form-bounded with respect to the Laplacian. Indeed, by the Hardy--Littlewood--Sobolev inequality $\check{\phi}*|x|^{-2}\in L^6(\R^3)$ and therefore, by H\"older's inequality,
$$
\int_{\R^d} |\check{\phi}*|x|^{-2}| |\psi|^2 \,dx \leq \|\check{\phi}*|x|^{-2}\|_6 \|\psi\|_{12/5}^2 \leq \|\check{\phi}*|x|^{-2}\|_6 \|\psi\|_6^{1/2} \|\psi\|^{3/2} \,. 
$$
By Sobolev's inequality we conclude that there is a $C$ such that for every $\epsilon>0$,
$$
\int_{\R^d} |\check{\phi}*|x|^{-2}| |\psi|^2 \,dx \leq \epsilon(\psi,p^2\psi) + C\epsilon^{-1/3} \|\check{\phi}*|x|^{-2}\|_6^{4/3} \|\psi\|^2 \,.
$$
Thus, $\check{\phi}*|x|^{-2}$ is infinitesimally form-bounded with respect to $p^2$ and we have
\begin{equation}
\label{eq:formbddh0}
H_\phi \geq (1-\epsilon)p^2 - C_\epsilon
\quad\text{and}\quad
H_\phi \leq (1+\epsilon)p^2 + C_\epsilon
\end{equation}
with $C_\epsilon= \epsilon^{-1/3} \|\check{\phi}*|x|^{-2}\|_6^{4/3} +\|\phi\|^2$. These two bounds imply (almost) conservation of the kinetic energy.

\begin{lemma}\label{energyconsh0}
If $\phi\in L^2(\R^3)$, then
$$
\sup_{t\in\R} \| (p^2+1)^{1/2} e^{-iH_\phi t} (p^2+1)^{-1/2} \| <\infty \,.
$$
\end{lemma}

\begin{proof}
For $\psi\in H^1(\R^3)$, by \eqref{eq:formbddh0},
\begin{align*}
\||p| e^{-iH_\phi t} \psi \|^2 & \leq (1-\epsilon)^{-1} \|( H_\phi +C_\epsilon)^{1/2} e^{-iH_\phi t}\psi\|^2 = (1-\epsilon)^{-1} \|( H_\phi+C_\epsilon)^{1/2} \psi\|^2 \\
& \leq (1-\epsilon)^{-1} \|((1+\epsilon)p^2 + 2C_\epsilon)^{1/2} \psi\|^2
\end{align*}
This clearly implies the assertion.
\end{proof}


\subsection{Creation and annihilation operators}

In this section we consider operators of the form
$$
a(e^{ik\cdot x} f) = \int_{\R^3} dk\, e^{-ik\cdot x} \overline{f(k)} a_k
\quad\text{and}\quad
a^*(e^{ik\cdot x} f) = \int_{\R^3} dk\, e^{ik\cdot x} f(k) a_k^*
$$
acting in $L^2(\R^3)\otimes\mathcal F$, where $f$ is a given function in $L^2(\R^3)$. We shall show that these operators can be bounded in terms of the square root of the number of particles operator $\mathcal N$, see \eqref{eq:nop}. We have

\begin{lemma}\label{cran}
Let $f\in L^2(\R^3)$. Then
$$
\left\|a(e^{ik\cdot x} f)\Psi\right\| \leq \|f\| \left\|\mathcal N^{1/2} \Psi\right\|
\quad\text{and}\quad
\left\|a^*(e^{ik\cdot x} f)\Psi\right\| \leq \|f\| \left\|(\mathcal N+\alpha^{-2})^{1/2} \Psi\right\| \,.
$$
Moreover,
$$
\left\|\mathcal N^{1/2} a(e^{ik\cdot x}f)\Psi\right\| \leq \|f\| \left\|\left(\mathcal N\left(\mathcal N-\alpha^{-2}\right)\right)^{1/2}\Psi\right\|
$$
and
$$
\left\|\mathcal N^{1/2} a^*(e^{ik\cdot x}f)\Psi\right\| \leq \|f\| \left\|\left(\mathcal N+\alpha^{-2}\right)\Psi\right\| \,.
$$
\end{lemma}

The proof is well-known and elementary, but we include it for the sake of completeness.

\begin{proof}
The first inequality follows from
$$
\left\|a(e^{ik\cdot x} f)\Psi\right\| \leq \int_{\R^3}dk\, |f(k)| \|a_k \Psi\| \leq \|f\| \left\|\mathcal N^{1/2} \Psi\right\| \,.
$$
To prove the second one, we use the intertwining relations 
\begin{equation}\label{eq:inter} a(f) h(\mathcal{N}) = h(\mathcal{N} +\alpha^{-2}) a(f) \quad \text{ and } \quad a^* (f) h(\mathcal{N}+\alpha^{-2}) = h (\mathcal{N}) a^* (f) \,,
\end{equation}
which hold for any function $h: \alpha^{-2} \mathbb{N}_0 \to \alpha^{-2} \mathbb{N}_0$ and follow from the canonical 
commutation relations (\ref{eq:ccr}).  These relations (together with the first bound in the lemma) imply that \[ \| a^*(e^{ik\cdot x} f) (\mathcal{N} + \alpha^{-2})^{-1/2} \| = \|\mathcal{N}^{-1/2} a^* (e^{ik\cdot x}f) \| = \| a(e^{ik \cdot x} f) \, \mathcal{N}^{-1/2} \| \leq \|f \| \, . \]
The third and fourth bound follow from the first two and again from the intertwining relations (\ref{eq:inter}). 
\end{proof}

We shall need the following corollary later in our proof.

\begin{corollary}\label{pcr}
Let $(1+|k|)f\in L^2(\R^3)$. Then
$$
\left\|\left(p^2+1\right)^{1/2} \mathcal N^{1/2} a^*(e^{ik\cdot x} f) \Psi\right\| \leq C \left\|\left(1+|k|\right)f\right\| \left\|\left(p^2+1\right)^{1/2}\left(\mathcal N+\alpha^{-2}\right)\Psi\right\| \,.
$$ 
\end{corollary}

\begin{proof}
We write
\begin{align*}
\left\|\left(p^2+1\right)^{1/2} \mathcal N^{1/2} a^*(e^{ik\cdot x} f) \Psi\right\|^2 = \|\mathcal N^{1/2} a^*(e^{ik\cdot x} f)\Psi\|^2 + \sum_{j=1}^3 \left\|\mathcal N^{1/2} p_j a^*(e^{ik\cdot x}f)\Psi\right\|^2 \,.
\end{align*}
The bound for $\mathcal N^{1/2} a^*(e^{ik\cdot x} f)\Psi$ follows from the second part of Lemma \ref{cran}. To bound the remaining terms we observe that
\begin{align*}
p_j a^*(e^{ik\cdot x}f) = a^*(e^{ik\cdot x}f)p_j + [p_j, a^*(e^{ik\cdot x}f)] = a^*(e^{ik\cdot x}f)p_j + a^*(e^{ik\cdot x}k_jf) \,.
\end{align*}
Thus,
\begin{align*}
\left\|\mathcal N^{1/2} p_j a^*(e^{ik\cdot x}f)\Psi\right\|
& \leq \left\|\mathcal N^{1/2} a^*(e^{ik\cdot x}f) p_j \Psi\right\|
+ \left\|\mathcal N^{1/2} a^*(e^{ik\cdot x}k_j f)\Psi\right\|
\end{align*}
and the assertion follows again from the second part of Lemma \ref{cran}.
\end{proof}


\subsection{The operator $H$}

Our next goal is to prove that the operator $H$ is lower semi-bounded. Indeed, we shall show that $H$ differs form $p^2+\mathcal N$ by terms which are infinitesimally form bounded with respect to $p^2+\mathcal N$. We begin with

\begin{lemma}\label{interact}
If $f\in L^2(\R^3)$ and $\epsilon>0$, then
$$
a(e^{ikx} f)+a^*(e^{ikx}f) \leq \epsilon \mathcal N + \epsilon^{-1}\|f\|^2 \,.
$$
\end{lemma}

Clearly, replacing $f$ by $-f$, we also obtain
$$
a(e^{ikx} f)+a^*(e^{ikx}f) \geq -\epsilon \mathcal N - \epsilon^{-1}\|f\|^2 \,.
$$

\begin{proof}
We have
\begin{align*}
0 & \leq \int_{\R^3} dk \left( \epsilon^{1/2} a_k^* - \epsilon^{-1/2}e^{-ikx} \overline{f(k)}\right) \left( \epsilon^{1/2} a_k - \epsilon^{-1/2}e^{ikx} f(k)\right) \\
& = \epsilon\mathcal N + \epsilon^{-1}\|f\|^2 - a^*(e^{ikx}f)-a(e^{ikx}f) \,,
\end{align*}
which implies the assertion.
\end{proof}

The following lemma is considerably more involved. It allows one to deal with the non-$L^2$ tail of $|k|^{-1}$ and is essentially due to Lieb and Yamazaki \cite{LiYa}.

\begin{lemma}\label{interactly}
If $|k|^{-1}f\in L^2(\R^3)$ and $\epsilon>0$, then
$$
a(e^{ikx} f)+a^*(e^{ikx}f) \leq \epsilon p^2 + 2\epsilon^{-1}\||k|^{-1}f\|^2 \left( 2\mathcal N + \alpha^{-2}\right) \,.
$$
\end{lemma}

Again, replacing $f$ by $-f$, we obtain
$$
a(e^{ikx} f)+a^*(e^{ikx}f) \geq -\epsilon p^2 - 2\epsilon^{-1}\||k|^{-1}f\|^2 \left( 2\mathcal N + \alpha^{-2}\right) \,.
$$

\begin{proof}
For $j=1,2,3$, we introduce
$$
Z_j= \int_{\R^d} dk \, \frac{k_j}{k^2} \, e^{-ikx} \, \overline{f(k)} a_k
$$
and write
\begin{align*}
a(e^{ikx} f)+a^*(e^{ikx}f) = \sum_{j=1}^3 [Z_j-Z_j^*,p_j] = \sum_{j=1}^3 \left(\left(Z_j-Z_j^*\right)p_j + p_j \left(Z_j^* -Z_j\right) \right) \,.  
\end{align*}
We bound, for every $j$,
\begin{align*}
\left(Z_j-Z_j^*\right)p_j + p_j \left(Z_j^* -Z_j\right) 
& \leq \epsilon p_j^2 + \epsilon^{-1} \left(Z_j-Z_j^*\right) \left(Z_j^*-Z_j\right) \\
& \leq \epsilon p_j^2 + 2 \epsilon^{-1} \left(Z_j^*Z_j + Z_j Z_j^* \right) \\
& = \epsilon p_j^2 + 2\epsilon^{-1} \left(2 Z_j^*Z_j + [Z_j, Z_j^*] \right)
\,.
\end{align*}
It remains to bound the last two terms. For every $\Psi$, we have, by Cauchy--Schwarz,
\begin{align*}
\langle \Psi, \sum_{j=1}^3 Z_j^*Z_j \Psi \rangle & = \iint_{\R^3\times\R^3} \frac{dk'}{k'^2} \frac{dk}{k^2} k'\cdot k \ f(k') \overline{f(k)} \langle \Psi,a_{k'}^* e^{i(k'-k)\cdot x} a_k \Psi\rangle \\
& \leq \left( \int_{\R^3} \frac{dk}{|k|} |f(k)| \|a_k \Psi\| \right)^2 \\
& \leq \||k|^{-1}f\|^2 \langle\Psi,\mathcal N\Psi\rangle \,.
\end{align*}
On the other hand, because of the commutation relations we have
\begin{align*}
\sum_{j=1}^3 [Z_j, Z_j^*] = \iint_{\R^3\times\R^3} \frac{dk'}{k'^2} \frac{dk}{k^2} k'\cdot k \ \overline{f(k')} f(k) e^{-i(k'-k)\cdot x} [a_{k'},  a_k^*] = \alpha^{-2} \||k|^{-1}f\|^2 \,.
\end{align*}
This concludes the proof of the lemma.
\end{proof}

We are now in position to prove form-boundedness. Given a number $\Lambda>0$ to be specified later, we decompose
\begin{equation}
\label{eq:hdecomp}
H=H_\phi + A + B + B^* \,,
\end{equation}
where
\begin{align*}
A = \mathcal N + a(\phi) + a^*(\phi) + \int_{|k|<\Lambda} \frac{dk}{|k|} \left( e^{ik\cdot x} a_k + e^{-ik\cdot x} a_k^* \right)
\end{align*}
and
$$
B = \int_{|k|>\Lambda} \frac{dk}{|k|} e^{ik\cdot x} a_k \,.
$$
For any choice of $\Lambda>0$, Lemma \ref{interact} implies that $A-\mathcal N$ is infinitesimally form bounded with respect to $\mathcal N$.

We claim that for any $\epsilon>0$ there is a $\Lambda>0$ such that $B+B^*$ is form bounded with respect to $p^2 + \mathcal N$ with form bound $\epsilon$. Indeed, this follows from Lemma \ref{interactly} by choosing $\Lambda$ so large that
$$
4\epsilon^{-1} \left\| |k|^{-2} \chi_{\{|k|>\Lambda\}} \right\|^2 =\epsilon \,. 
$$

This argument shows that for every $\epsilon>0$ there is a $C_\epsilon$ and a $\Lambda$ such that
$$
H \geq (1-\epsilon) (p^2 +\mathcal N) - C_\epsilon
\quad\text{and}\quad
H \leq (1+\epsilon) (p^2 +\mathcal N) + C_\epsilon \,.
$$
The constant $C_\epsilon$ depends on $\alpha$ through the use of Lemma \ref{interactly}, but it is uniformly bounded for $\alpha\geq \alpha_0$. Thus, by the same argument as in Lemma \ref{energyconsh0} we obtain

\begin{lemma}\label{energyconsh}
If $\phi\in L^2(\R^3)$ and $\alpha_0>0$, then
$$
\sup_{\alpha\geq \alpha_0}\ \sup_{t\in\R} \left\| \left(p^2 +\mathcal N+1\right)^{1/2} e^{-iHt} \left(p^2+\mathcal N+1\right)^{-1/2} \right\| <\infty \,.
$$
\end{lemma}


\section{Proof of Theorem \ref{main2}}

We shall prove Theorem \ref{main2} by a Gronwall-type argument. More precisely, we shall prove the following proposition.

\begin{proposition}\label{mainprop}
Let $\Psi$ be as in Theorem \ref{main2}. Then
$$
\frac{d}{dt} \left\| \left(e^{-iHt}-e^{-iH_\phi t}\right)\Psi \right\|^2 = f(t) + g(t) \,,
$$
where
$$
f(t) \leq C M \alpha^{-1} \left\| \left(e^{-iHt}-e^{-iH_\phi t}\right)\Psi \right\|
$$
and, for all $T\geq 0$,
$$
\int_0^T dt\, g(t) \leq C M^2 \alpha^{-2} T \,.
$$
Here, $C$ depends only on $\alpha_0$ and an upper bound on $\|\phi\|$.
\end{proposition}

\begin{proof}[Proof of Theorem \ref{main2} given Proposition \ref{mainprop}]
It suffices to consider times $T\geq 0$. Then
$$
A(T) := \left\| \left(e^{-iHT}-e^{-iH_\phi T}\right)\Psi \right\|^2 = \int_0^T dt\, f(t) + \int_0^T dt\, g(t) \,.
$$
According to Proposition \ref{mainprop} we have $f(t) \leq C M^2 \alpha^{-2} + CA(t)$. This, together with the bound on the integral of $g$, implies
$$
A(T) \leq 2C M^2 \alpha^{-2}T + C \int_0^T dt\, A(t) \,.
$$
Thus,
$$
\left(A(T)+ 2M^2 \alpha^{-2}\right) \leq 2M^2 \alpha^{-2} + C \int_0^T dt\, (A(t)+2M^2 \alpha^{-2})
$$
and, by Gronwall's inequality,
$$
A(t)+2M^2 \alpha^{-2} \leq 2M^2 \alpha^{-2} e^{Ct} \,.
$$
This is inequality \eqref{eq:mainequiv} which, as explained before, is equivalent to the inequality stated in Theorem \ref{main2}.
\end{proof}

It remains to prove Proposition \ref{mainprop}, and so we differentiate
\begin{align*}
\frac{d}{dt} \left\| \left(e^{-iHt}-e^{-iH_\phi t}\right)\Psi \right\|^2
& = 2\im \langle e^{-iHt}\Psi,(H-H_\phi) e^{-iH_\phi t}\Psi\rangle \\
& = 2\im \langle \left( e^{-iHt}-e^{-iH_\phi t}\right)\Psi,(H-H_\phi) e^{-iH_\phi t}\Psi\rangle \\
& = f_1(t)+f_2(t)+h(t)\,.
\end{align*}
In the middle equality we used the fact that $H$ and $H_\phi$ are self-adjoint. The functions $f_1$, $f_2$ and $h$ are defined by
\begin{align*}
f_1(t) &= 2\im \langle \left( e^{-iHt}-e^{-iH_\phi t}\right)\Psi,A e^{-iH_\phi t}\Psi\rangle \,,\\
f_2(t) &= 2\im \langle \left( e^{-iHt}-e^{-iH_\phi t}\right)\Psi,B e^{-iH_\phi t}\Psi\rangle \,, \\
h(t) &= 2\im \langle \left( e^{-iHt}-e^{-iH_\phi t}\right)\Psi,B^* e^{-iH_\phi t}\Psi\rangle
\end{align*}
in terms of the decomposition $H=H_\phi+A+B+B^*$ from \eqref{eq:hdecomp}. As we will see below, the functions $f_1$ and $f_2$ contribute to the $f$-piece in Proposition \ref{mainprop}, whereas $h$ will be further decomposed into an $f$-piece and a $g$-piece.

In the decomposition above, the cut-off value $\Lambda$ is fixed and we do not make it explicit in our bounds. Also, we do not indicate the dependence of the constants on $\phi$ (and its $L^2$-norm) and $\alpha_0$. As a final preliminary, let us note that the a-priori bounds \eqref{eq:apriori} implies
\begin{equation}
\label{eq:apriori1}
\left\|\left(p^2+1\right)^{1/2} \mathcal N^{1/2} \Psi \right\| \leq M\alpha^{-1} \,.
\end{equation}
Indeed, this follows by the Cauchy--Schwarz inequality, since $\|(p^2+1)^{1/2}\Psi\|\leq M$ and $\|(p^2+1)^{1/2}\mathcal N\Psi\|\leq M\alpha^{-2}$. Moreover, by \eqref{eq:apriori}
\begin{equation}
\label{eq:apriori2}
\left\| \mathcal N \Psi \right\| \leq C M\alpha^{-1} \,.
\end{equation}
with $C=\alpha_0^{-1}$.


\subsection{Bound on $f_1$}

It is an easy consequence of Lemma \ref{cran} that
$$
\| A\xi\| \leq C \left\| \left(\mathcal N + \alpha^{-1}\right)\xi\right\|
\qquad\text{for all}\ \xi \,,
$$
and, thus,
\begin{align*}
|f_1(t)| & \leq 2 \left\| \left( e^{-iHt}-e^{-iH_\phi t}\right)\Psi\right\| \left\|A e^{-iH_\phi t}\Psi\right\| \\
& \leq 2C \left\| \left( e^{-iHt}-e^{-iH_\phi t}\right)\Psi\right\| \left\|\left(\mathcal N + \alpha^{-1}\right) \Psi\right\| \\
& \leq C' M \alpha^{-1} \left\| \left( e^{-iHt}-e^{-iH_\phi t}\right)\Psi\right\| \,.
\end{align*}
Here we also used \eqref{eq:apriori2} and the fact that $\mathcal N$ commutes with $H_\phi$. This bound on $f_1$ is already of the form required for the application of Proposition \ref{mainprop}.


\subsection{Bound on $f_2$}

To estimate $f_2$ we make use of the following lemma.

\begin{lemma}\label{lemmab}
We have, with a constant depending only on $\Lambda$,
$$
\|B\xi\| \leq C \left\| \left(p^2+1\right)^{1/2} \mathcal N^{1/2}\xi\right\|
$$
and
$$
\left\|\left(\mathcal N+\alpha^{-2}\right)^{-1/2}B\xi\right\| \leq C \left\| \left(p^2+1\right)^{1/2}\xi\right\| \,.
$$
\end{lemma}

\begin{proof}
If we describe the electron in momentum space, then $B\xi$ for $\xi=(\xi^{(0)},\xi^{(1)},\ldots)$ is given by
$$
\left(B\xi\right)^{(n)}(p,k_1,\ldots,k_n) = \sqrt\alpha\, \sqrt{n+1} \int_{|k|>\Lambda} \frac{dk}{|k|} \xi^{(n+1)}(p+k,\alpha k,k_1,\ldots,k_n) \,. 
$$
This follows from the standard representation of $a(f)$ together with the rescaling explained in the appendix. By Cauchy--Schwarz,
\begin{align*}
\left\|B\xi\right\|^2 & = \alpha \sum_{n=0}^\infty (n+1) \int_{\R^3}dp\, \int_{\R^{3n}}d{\bf k} \left| \int_{|k|>\Lambda} \frac{dk}{|k|} \xi^{(n+1)}(p+k,\alpha k,{\bf k}) \right|^2 \\
& = \alpha \sum_{n=0}^\infty (n+1) \int_{\R^3}dp\, \int_{\R^{3n}}d{\bf k} \\ 
&\hspace{2cm} \times \iint_{|k|>\Lambda,|k'|>\Lambda} \frac{dk'}{|k'|}\frac{dk}{|k|} \overline{\xi^{(n+1)}(p+k',\alpha k',{\bf k})} \xi^{(n+1)}(p+k,\alpha k,{\bf k}) \\
& \leq \alpha \sum_{n=0}^\infty (n+1) \int_{\R^3}dp\, \int_{\R^{3n}}d{\bf k} \\ 
&\hspace{2cm} \times \iint_{|k|>\Lambda,|k'|>\Lambda} dk'\,dk\, \frac{1+(p+k)^2}{k'^2(1+(p+k')^2)} |\xi^{(n+1)}(p+k,\alpha k,{\bf k})|^2 \\
& \leq C \alpha \sum_{n=0}^\infty (n+1) \int_{\R^3}dp\, \int_{\R^{3n}}d{\bf k} \int_{\R^3} dk\, (1+(p+k)^2) |\xi^{(n+1)}(p+k,\alpha k,{\bf k})|^2 \\
& = C \alpha^{-2} \sum_{n=0}^\infty (n+1) \int_{\R^3}dp\, \int_{\R^{3n}}d{\bf k} \int_{\R^3} d\tilde{k}\, (1+p^2) |\xi^{(n+1)}(p,\tilde{k},{\bf k})|^2 \\
& = C \left\|\mathcal N^{1/2} (1+p^2)^{1/2}\xi\right\|^2
\end{align*}
with
$$
C= \sup_{p\in\R^3} \int_{|k'|>\Lambda} \frac{dk'}{{k'}^2(1+(p+k')^2)} <\infty \,.
$$
This proves the first bound in the lemma. The second one is proved similarly and we omit the details.
\end{proof}

Using this lemma, we bound
\begin{align*}
|f_2(t)| & \leq 2 \left\| \left( e^{-iHt}-e^{-iH_\phi t}\right)\Psi\right\| \left\|B e^{-iH_\phi t}\Psi\right\| \\
& \leq 2C \left\| \left( e^{-iHt}-e^{-iH_\phi t}\right)\Psi\right\| \left\|(p^2+1)^{1/2} \mathcal N^{1/2} e^{-iH_\phi t} \Psi\right\| \,.
\end{align*}
Since $\mathcal N$ commutes with $H_\phi$, by means of the energy conservation lemma \ref{energyconsh0} and by \eqref{eq:apriori1} we find
$$
\left\|(p^2+1)^{1/2} \mathcal N^{1/2} e^{-iH_\phi t} \Psi\right\|
\leq C' \left\|(p^2+1)^{1/2}  \mathcal N^{1/2} \Psi\right\|
\leq C' M \alpha^{-1} \,.
$$
Thus,
$$
|f_2(t)| \leq 2CC' M\alpha^{-1} \left\| \left( e^{-iHt}-e^{-iH_\phi t}\right)\Psi\right\| \,,
$$
which is a bound of the form required for Proposition \ref{mainprop}.


\subsection{Decomposition of $h$}

It remains to deal with the term $h$, which involves the operator $B^*$. We split this operator as follows,
\begin{align*}
B^* & = \int_{|k|>\Lambda} \frac{dk}{|k|^3} \left[ k\cdot p,e^{ik\cdot x}\right] a_k^* \\
&= \int_{|k|>\Lambda} \frac{dk}{|k|^3} \left( k\cdot p e^{ik\cdot x} + e^{ik\cdot x} k\cdot p\right) a_k^*
-2 \int_{|k|>\Lambda} \frac{dk}{|k|^3} e^{ik\cdot x} k\cdot p \ a_k^* \\
&= \left[H_\phi, a^*( e^{ikx} |k|^{-3} \chi_{\{|k|>\Lambda\}})\right] -2 \int_{|k|>\Lambda} \frac{dk}{|k|^3} e^{ik\cdot x} k\cdot p \ a_k^* \,.
\end{align*}
Accordingly, we decompose
$$
h(t) = f_3(t) + g(t) \,,
$$
where
$$
f_3(t) = - 4\im \left\langle \left( e^{-iHt}-e^{-iH_\phi t}\right)\Psi,\int_{|k|>\Lambda} \frac{dk}{|k|^3} e^{ik\cdot x} k\cdot p \ a_k^* \ e^{-iH_\phi t} \Psi\right\rangle
$$
and
$$
g(t) = 2\im \left\langle \left( e^{-iHt}-e^{-iH_\phi t}\right)\Psi,\left[H_\phi, a^*( e^{ikx} |k|^{-3} \chi_{\{|k|>\Lambda\}})\right] e^{-iH_\phi t} \Psi\right\rangle \,.
$$


\subsection{Bound on $f_3$}

We bound
\begin{align*}
|f_3(t)| & \leq 4 \left\| \left( e^{-iHt}-e^{-iH_\phi t}\right)\Psi\right\| \left\|\int_{|k|>\Lambda} \frac{dk}{|k|^3} e^{ik\cdot x} k\cdot p \ a_k^* \ e^{-iH_\phi t} \Psi\right\| \\
& \leq 4 \left\| \left( e^{-iHt}-e^{-iH_\phi t}\right)\Psi\right\| \sum_{j=1}^3
\left\|a^*(e^{ik\cdot x}k_j |k|^{-3} \chi_{\{|k|>\Lambda\}}) p_j \ e^{-iH_\phi t} \Psi\right\| \,.
\end{align*}
According to Lemma \ref{cran} and energy conservation, Lemma \ref{energyconsh0}, we have
\begin{align*}
\Big \|a^*(e^{ik\cdot x}k_j |k|^{-3} &\chi_{\{|k|>\Lambda\}}) p_j \ e^{-iH_\phi t} \Psi \Big\|
\\ & \leq \| k_j |k|^{-3} \chi_{\{|k|>\Lambda\}} \| \left\|\left( \mathcal N+\alpha^{-2}\right)^{1/2} p_j \ e^{-iH_\phi t} \Psi\right\| \\
& = \| k_j |k|^{-3} \chi_{\{|k|>\Lambda\}} \| \left\|p_j \ e^{-iH_\phi t} \left( \mathcal N+\alpha^{-2}\right)^{1/2} \Psi\right\| \\
& \leq C \| k_j |k|^{-3} \chi_{\{|k|>\Lambda\}} \| \left\|(p^2+1)^{1/2} \left( \mathcal N+\alpha^{-2}\right)^{1/2} \Psi\right\| \\
& \leq \sqrt 2\, C \| k_j |k|^{-3} \chi_{\{|k|>\Lambda\}} \| M \alpha^{-1} \,.
\end{align*}
Here we used \eqref{eq:apriori1}. Thus, $f_3$ is bounded as required for Proposition \ref{mainprop}.


\subsection{Decomposition of the integral of $g$}

We want to use the fact that
$$
e^{iH_\phi t}\left[H_\phi, a^*( e^{ikx} |k|^{-3} \chi_{\{|k|>\Lambda\}})\right] e^{-iH_\phi t} = -i \frac{d}{dt} \left( e^{iH_\phi t} a^*( e^{ikx} |k|^{-3} \chi_{\{|k|>\Lambda\}}) e^{-iH_\phi t} \right) \,.
$$
This implies that
\begin{align*}
& \left( e^{iHt}-e^{iH_\phi t}\right) \left[H_\phi, a^*( e^{ikx} |k|^{-3} \chi_{\{|k|>\Lambda\}})\right] e^{-iH_\phi t} \\
& \quad = - i \left( e^{iHt}-e^{iH_\phi t}\right) e^{-iH_\phi t} \frac{d}{dt} \left( e^{iH_\phi t} a^*( e^{ikx} |k|^{-3} \chi_{\{|k|>\Lambda\}}) e^{-iH_\phi t} \right) \\
& \quad = \int_0^t ds\, e^{iHs} (H-H_\phi)e^{-iH_\phi s} \frac{d}{dt} \left( e^{iH_\phi t} a^*( e^{ikx} |k|^{-3} \chi_{\{|k|>\Lambda\}}) e^{-iH_\phi t} \right) \,.
\end{align*}
Integrating by parts, we find that
\begin{align*}
\int_0^T& dt\, g(t) \\ & = 2\im \int_0^Tdt\, \left\langle \int_0^t ds\, e^{iH_\phi s} (H-H_\phi) e^{-iHs}\Psi, \frac{d}{dt} e^{iH_\phi t} a^*( e^{ikx} |k|^{-3} \chi_{\{|k|>\Lambda\}}) e^{-iH_\phi t} \Psi \right\rangle \\
& = - 2\im \int_0^Tdt\, \left\langle e^{iH_\phi t} (H-H_\phi) e^{-iHt}\Psi, e^{iH_\phi t} a^*( e^{ikx} |k|^{-3} \chi_{\{|k|>\Lambda\}}) e^{-iH_\phi t} \Psi \right\rangle \\
& \quad + 2\im \left\langle \int_0^T ds\, e^{iH_\phi s} (H-H_\phi) e^{-iHs}\Psi, e^{iH_\phi T} a^*( e^{ikx} |k|^{-3} \chi_{\{|k|>\Lambda\}}) e^{-iH_\phi T} \Psi \right\rangle \\
& = - 2\im \int_0^Tdt\, \left\langle (H-H_\phi) e^{-iHt}\Psi, \Psi(t) \right\rangle
\end{align*}
with
$$
\Psi(t) = \left( a^*( e^{ikx} |k|^{-3} \chi_{\{|k|>\Lambda\}}) - e^{iH_\phi (T-t)} a^*( e^{ikx} |k|^{-3} \chi_{\{|k|>\Lambda\}}) e^{-iH_\phi (T-t)} \right) e^{-iH_\phi t} \Psi \,.
$$
We decompose again $H=H_\phi +A+B +B^*$ as in \eqref{eq:hdecomp} and accordingly
$$
\int_0^Tdt\, g(t) = G_1(T) + G_2(T) + G_3(T) 
$$
with
\begin{align*}
G_1(T) & = -2\im \int_0^Tdt\, \left\langle A e^{-iHt}\Psi, \Psi(t) \right\rangle \,,\\
G_2(T) & = -2\im \int_0^Tdt\, \left\langle B e^{-iHt}\Psi, \Psi(t) \right\rangle \,,\\
G_3(T) & = -2\im \int_0^Tdt\, \left\langle B^* e^{-iHt}\Psi, \Psi(t) \right\rangle \,.
\end{align*}
It remains to bound these three terms.


\subsection{Bound on $G_1$}

If we write $A=\mathcal N+\tilde A$, we obtain from Lemma \ref{cran} that
$$
\left\|\tilde A\xi\right\| \leq C \left\|\left(\mathcal N+\alpha^{-2}\right)^{1/2} \xi\right\|
\qquad\text{for all}\ \xi \,.
$$
This allows us to bound
\begin{align*}
|G_1(T)| & \leq 2 \int_0^T dt\, \left( \|\mathcal N^{1/2} e^{-iHt}\Psi\| \|\mathcal N^{1/2} \Psi(t)\| + \|e^{-iHt}\Psi\| \|\tilde A \Psi(t)\| \right) \\
& \leq 2 \int_0^T dt\, \left( \|\mathcal N^{1/2} e^{-iHt}\Psi\| \|\mathcal N^{1/2} \Psi(t)\| + C M \left\| \left(\mathcal N+\alpha^{-2}\right)^{1/2} \Psi(t) \right\| \right) \,.
\end{align*}
According to energy conservation, Lemma \ref{energyconsh}, we have
$$
\|\mathcal N^{1/2} e^{-iHt}\Psi\| \leq C \|(p^2 + \mathcal N+1)^{1/2} \Psi\| \leq CM \,.
$$
Thus, it remains to bound the norm of $\Psi(t)$ and $\mathcal N^{1/2}\Psi(t)$. By Lemma \ref{cran},
\begin{align*}
\|\Psi(t)\| & \leq \| a^*( e^{ikx} |k|^{-3} \chi_{\{|k|>\Lambda\}}) e^{-iH_\phi t} \Psi\| + \| a^*( e^{ikx} |k|^{-3} \chi_{\{|k|>\Lambda\}}) e^{-iH_\phi T} \Psi\| \\
& \leq 2 \left\| |k|^{-3} \chi_{\{|k|^{-3}>\Lambda\}}\right\| \left\|\left(\mathcal N+\alpha^{-2}\right)^{1/2} \Psi\right\| \\
& \leq 2 \sqrt2 \left\| |k|^{-3} \chi_{\{|k|^{-3}>\Lambda\}}\right\| M \alpha^{-1} \,,
\end{align*}
where we used \eqref{eq:apriori1}. Moreover, again by Lemma \ref{cran},
\begin{align*}
\|\mathcal N^{1/2} \Psi(t)\|  \leq \; &\| \mathcal N^{1/2} a^*( e^{ikx} |k|^{-3} \chi_{\{|k|>\Lambda\}}) e^{-iH_\phi t} \Psi\| \\ &+ \| \mathcal N^{1/2} a^*( e^{ikx} |k|^{-3} \chi_{\{|k|>\Lambda\}}) e^{-iH_\phi T} \Psi\| \\
 \leq \; &2 \left\| |k|^{-3} \chi_{\{|k|^{-3}>\Lambda\}}\right\| \left\|\left(\mathcal N+\alpha^{-2}\right) \Psi\right\| \\
\leq \; &4 \left\| |k|^{-3} \chi_{\{|k|^{-3}>\Lambda\}}\right\| M \alpha^{-2} \,.
\end{align*}
Note that the previous two bounds also imply that
$$
\left\| \left(\mathcal N+\alpha^{-2}\right)^{1/2} \Psi(t) \right\|
\leq C' M\alpha^{-2} \,.
$$
Combining everything we infer that
$$
|G_1(T)| \leq C'' M \alpha^{-2} T \,,
$$
as required for Proposition \ref{mainprop}.


\subsection{Bound on $G_2$}

Using the second inequality in Lemma \ref{lemmab} we get
\begin{align*}
|G_2(T)| & \leq 2 \int_0^T dt\, \left\|\left(\mathcal N+\alpha^{-2}\right)^{-1/2} Be^{iHt}\Psi\right\| \left\|\left(\mathcal N+\alpha^{-2}\right)^{1/2} \Psi(t) \right\| \\
& \leq 2C \int_0^T dt\, \left\|\left(p^2+1\right)^{1/2} e^{iHt}\Psi\right\| \left\|\left(\mathcal N+\alpha^{-2}\right)^{1/2} \Psi(t) \right\| \,.
\end{align*}
By energy conservation, Lemma \ref{energyconsh}, we have
$$
\left\|\left(p^2+1\right)^{1/2} e^{iHt}\Psi\right\| \leq C \left\|\left(p^2+\mathcal N +1\right)^{1/2} \Psi\right\| \leq CM \,.
$$
This, together with the bound on $\left(\mathcal N+\alpha^{-2}\right)^{1/2} \Psi(t)$ that we derived when bounding $G_1$, yields a bound on $G_2$ of the desired form.


\subsection{Bound on $G_3$}

We bound, using the first inequality in Lemma \ref{lemmab},
\begin{align*}
|G_3(T)| & \leq 2 \int_0^Tdt\, \left\| e^{-iHt}\Psi \right\| \left\| B\Psi(t) \right\| \\
& \leq 2 C \int_0^Tdt\, \left\|\left(p^2+1\right)^{1/2}\mathcal N^{1/2} \Psi(t) \right\| \,. 
\end{align*}
By energy conservation, Lemma \ref{energyconsh0}, together with Corollary \ref{pcr} and the fact that $(1+|k|)|k|^{-3}\chi_{\{|k|>\Lambda\}}\in L^2$,
\begin{align*}
\left\|\left(p^2+1\right)^{1/2} \mathcal N^{1/2} \Psi(t) \right\|
& \leq \left\|\left(p^2+1\right)^{1/2} \mathcal N^{1/2} a^*( e^{ikx} |k|^{-3} \chi_{\{|k|>\Lambda\}}) e^{-iH_\phi t} \Psi \right\| \\
& \qquad + C \left\| \left(p^2+1\right)^{1/2} \mathcal N^{1/2} a^*( e^{ikx} |k|^{-3} \chi_{\{|k|>\Lambda\}}) e^{-iH_\phi T} \Psi \right\| \\
& \leq C' \left\|\left(p^2+1\right)^{1/2} \left(\mathcal N+\alpha^{-2}\right) e^{-iH_\phi t} \Psi \right\| \\
& \qquad + CC' \left\| \left(p^2+1\right)^{1/2} \left(\mathcal N+\alpha^{-2}\right) e^{-iH_\phi T} \Psi \right\| \\
& \leq C'' \left\|\left(p^2+1\right)^{1/2} \left(\mathcal N+\alpha^{-2}\right) \Psi \right\|\\
& \leq 2 C'' M \alpha^{-2} \,.
\end{align*}
Again this shows that $G_3$ is bounded as required for the application of Proposition \ref{mainprop}. The proof of Proposition \ref{mainprop} is now complete.


\appendix

\section{Strong coupling units}

In this appendix we briefly explain how $H^{\text{F}}_\alpha$ is related to the more traditional form of the Fr\"ohlich Hamiltonian
$$
p^2 + \sqrt\alpha \int_{\R^3} \frac{d k}{| k|} \left( e^{-i k\cdot  x}  a_{ k} + e^{i k\cdot  x} a_{ k}^* \right) + \int_{\R^3} dk\, a_k^* a_k \,,
$$
where now $a_k^*$ and $a_k$ satisfy
$$
[a_k,a_{k'}^*]=\delta(k-k') \,,
\qquad [a_k,a_{k'}]=[a_k^*,a_{k'}^*]=0
\qquad\text{for all}\ k,k'\in\R^3 \,.
$$
Let $\tilde x = \alpha x$, so that $\tilde p =\alpha^{-1} p$. Then the above operator is unitarily equivalent to
$$
\alpha^2 \tilde p^2 + \sqrt\alpha \int_{\R^3} \frac{d k}{| k|} \left( e^{-i \alpha^{-1} k\cdot  \tilde x}  a_{ k} + e^{i \alpha^{-1} k\cdot  \tilde x} a_{ k}^* \right) + \int_{\R^3} dk\, a_k^* a_k \,.
$$
By the change of variables $\tilde k = \alpha^{-1} k$ we can rewrite the operator as
\begin{align*}
& \alpha^2 \tilde p^2 + \alpha^{5/2} \int_{\R^3} \frac{d \tilde k}{|\tilde k|} \left( e^{-i \tilde k\cdot  \tilde x}  a_{\alpha \tilde k} + e^{i \tilde k\cdot  \tilde x} a_{\alpha\tilde k}^* \right) + \alpha^3 \int_{\R^3} dk\, a_{\alpha\tilde k}^* a_{\alpha\tilde k} \\
& = \alpha^2 \left( \tilde p^2 + \int_{\R^3} \frac{d \tilde k}{|\tilde k|} \left( e^{-i \tilde k\cdot  \tilde x}  \left(\alpha^{1/2} a_{\alpha \tilde k}\right) + e^{i \tilde k\cdot  \tilde x} \left(\alpha^{1/2} a_{\alpha\tilde k}^*\right) \right) \right. \\ &\hspace{7cm} \left. + \int_{\R^3} dk\, \left( \alpha^{1/2} a_{\alpha\tilde k}\right)^* \left(\alpha^{1/2} a_{\alpha\tilde k} \right) \right) \,.
\end{align*}
Defining $\tilde a_{\tilde k} = \alpha^{1/2} a_{\alpha\tilde k}$ we find the commutation relations
$$
[\tilde a_{\tilde k},\tilde a_{\tilde k'}^*]=\alpha^{-2} \delta(\tilde k-\tilde k') \,,
\qquad [\tilde a_{\tilde k},\tilde a_{\tilde k'}]=[\tilde a_{\tilde k}^*,\tilde a_{\tilde k'}^*]=0
\qquad\text{for all}\ \tilde k,\tilde k'\in\R^3 \,.
$$
Thus, we have obtained the Hamiltonian $\alpha^2 H^{\text{F}}_\alpha$.


\bibliographystyle{amsalpha}

\begin{thebibliography}{29}

\bibitem{BeNiRuSo} P. Bechouche, J. Nieto, E. Ruiz Arriola, J. Soler, \textit{On the time evolution of the mean-field polaron}. J. Math. Phys. \textbf{41} (2000), no. 7, 4293--4312.

\bibitem{DeAl} J. T. Devreese, A. S. Alexandrov, \textit{Fr\"ohlich polaron and bipolaron: recent developments}. Rep. Prog. Phys. \textbf{72} (2009), no. 6: 066501.

\bibitem{DoVa} M. Donsker, S. R. S. Varadhan, \textit{Asymptotics for the polaron}. Comm. Pure Appl. Math. \textbf{36} (1983), 505--528.

\bibitem{Fe} R. P. Feynman, {\it Slow electrons in a polar crystal}. Phys. Rev. {\bf 97} (1955), 660--665.

\bibitem{Fr} H. Fr\"ohlich, \textit{Theory of electrical breakdown in ionic crystals}. Proc. R. Soc. Lond. A \textbf{160} (1937), 230--241.

\bibitem{Gu} M. Gurari, \textit{Self-energy of slow electrons in polar materials}. Phil. Mag. Ser. 7 \textbf{44}:350 (1953), 329--336.

\bibitem{LaPe} L. D. Landau, S. I. Pekar, Zh. Eksp. Teor. Fiz. \textbf{18} (1948), 419.

\bibitem{LePi} T-D. Lee, D. Pines, \textit{The motion of slow electrons in polar crystals}. Phys. Rev. \textbf{88} (1952), 960--961.

\bibitem{LeLoPi} T-D. Lee, F. Low, D. Pines, \textit{The motion of slow electrons in a polar crystal}. Phys. Rev. \textbf{90} (1953), 297--302.

\bibitem{LiTh} E. H. Lieb, L. E. Thomas, \textit{Exact ground state energy of the strong-coupling polaron}. Comm. Math. Phys. \textbf{183}, no. 3, 511--519 (1997). Erratum: \textit{ibid.} \textbf{188},  no. 2, 499--500 (1997).

\bibitem{LiYa} E. H. Lieb, K. Yamazaki, \textit{Ground-state energy and effective mass of the polaron}. Phys. Rev. \textbf{111}, 728--722 (1958).

\bibitem{Pe0} S. I. Pekar, Zh. Eksp. Teor. Fiz. \textbf{16} (1946), 335.

\bibitem{Pe} S. I. Pekar, \textit{Research in electron theory of crystals}.
(Russian edition 1951, German edition 1954) US Atomic Energy Commission, AEC-tr-555, Washington, DC, 1963.

\end{thebibliography}

\end{document}